\documentclass{amsart}

\usepackage{amssymb,amsmath}

\newtheorem{theorem}{Theorem}[section]

\newtheorem{lemma}[theorem]{Lemma}

\begin{document}

\title{The price of fairness for a small number of indivisible items}
\author{Sascha Kurz}
\address{Business Mathematics, University of Bayreuth, 95440 Bayreuth, sascha.kurz@uni-bayreuth.de}

\maketitle

\begin{abstract}
  Incorporating fairness criteria in optimization problems comes at a certain cost, which is measured by the so-called price of fairness. Here 
we consider the allocation of indivisible goods. For envy-freeness as fairness criterion it is known from literature that the price of fairness can
increase linearly in terms of the number of agents. For the constructive lower bound a quadratic number of items was used. In practice this might be 
inadequately large. So we introduce the price of fairness in terms of both the number of agents and items, i.e., key parameters which generally may
be considered as common and available knowledge. It turned out that the price of fairness increases sublinear if the number of items is not too 
much larger than the number of agents. For the special case of coincide of both counts exact asymptotics could be determined. Additionally 
an efficient integer programming formulation is given.
\end{abstract}  

\section{Introduction}

Fair division, i.e., the the problem of dividing a set of goods between several agents, is studied since ancient times, see
e.g.\ \cite{brams1996fair}. As argued by Bertsimas et al. \cite{bertsimas2011price}, harming a certain fairness criterion 
can cause the situation that a globally optimal solution is not implementable by self-interested agents. So, several authors 
have studied the price of fairness as a measurement of the costs of ensuring a certain kind of fairness among the agents.

Here we consider the allocation of $m$ indivisible items among $n$ agents with respect to the fairness criterion of 
envy-freeness. Additionally we assume additive utility functions summing up to one for all agents. This setting and other 
fairness criteria and types of items have been studied, see e.g.\ \cite{caragiannis2012efficiency}.

While our theoretical setting is rather narrow, our contribution lies in highlighting that the number of items has 
a significant impact on the price of fairness. In our setting the price of fairness can be as large as $\Theta(n)$ if 
we allow a large number of items. If the number of items is restrict to a small number, compared to the number of agents, 
then it turns out that the price of fairness is $\Theta(\sqrt{n})$. Even the smallest possible case, 
admitting envy-free allocations, $m=n$ is far from being innocent. Nevertheless we determine its exact value for all $n$ up to a 
constant and give a fast-to-solve ILP formulation.

\section{Basic notation and definitions}
Let $\mathcal{J}=\{1,\dots,n\}$ be a set of agents and $\mathcal{I}=\{1,\dots,m\}$ be a set of indivisible items. Each agent $j\in\mathcal{J}$ 
has a non-negative and additive utility function $u_j$ over the subsets of $\mathcal{I}$ with $u_j(\emptyset)=0$ and $u_j(\mathcal{I})=1$. 
An allocation $\mathcal{A}=(A_1,\dots,A_n)$ is a partition of $\mathcal{I}$ meaning that the elements of $A_j$ are allocated to agent $j$. 
We call an allocation envy-free, if we have $u_{j}(A_j)\ge u_j(A_{j'})$ for all $j,j'\in\mathcal{J}$, i.e., no agent evaluates 
a share of one of the other agents higher than his own share. Depending on the utility functions there may be no envy-free allocation at all, 
consider e.g.\ $u_j(\{1\})=1$ and $u_j(\{i\})=0$ for all $i\neq 1$. As a global evaluation of an allocation we use
the sum of the agents utilities, i.e., $u(\mathcal{A})=\sum_{j=1}^n u_j(A_j)$. By $\mathcal{A}^\star$ we denote an allocation maximizing 
the global utility $u$ and similarly by $\mathcal{A}^\star_f$ we denote an envy-free allocation, if exists, maximizing $u$. With this the price 
of envy-freeness for $n$ agents $p_{envy}(n)$ is defined as the supremum of $u(\mathcal{A}^\star)/u(\mathcal{A}^\star_f)$. Obviously we have 
$p_{envy}(1)=1$. For $n>1$ the authors of \cite{caragiannis2012efficiency} have shown $\frac{3n+7}{9}\le p_{envy}(n)\le n-\frac{1}{2}$. 
Besides $p_{envy}(2)=\frac{3}{2}$ no exact value is known.

The construction for the lower bound of $p_{envy}(n)$ uses $\Omega(n^2)$ items so that one can ask if the price of fairness decreases if the number 
of items is restricted to a sub-quadratic number of items, which seems to be more reasonable in practice. So we define $p_{envy}(n,m)$ 
as the supremum of $u(\mathcal{A}^\star)/u(\mathcal{A}^\star_f)$, where to number of items equals $m$. In any envy-free allocation we have $u_j(A_j)\ge \frac{1}{n}$ since 
otherwise $u_j(\mathcal{I})<1$. Thus $\left|A_j\right|\ge 1$ so that we can assume $m\ge n$. The first case $m=n$ is studied in the next 
section. Obviously we have $p_{envy}(n,m)\le p_{envy}(n,m+1)\le p_{envy}(n)$ for all $m\ge n\ge 1$.

\section{The smallest case: One item per agent}
As an abbreviation we use $x_{ij}=u_j(\{i\})$ for all $1\le i,j\le n$. The maximum utility $u(\mathcal{A}^\star)$ can be easily 
determined as $\sum_{i=1}^n\max_{j} x_{ij}$ in linear time. As argued before in any envy-free allocation each agent is assigned exactly 
one item. W.l.o.g.\  we assume that item $j$ as assigned to agent $j$ for all $1\le i\le n$, i.e., we have $x_{jj}\ge x_{ij}$ for 
all $1\le i,j\le n$. Using a matching algorithm the existence of an envy-free allocation for $m=n$ can be checked in polynomial time. 
For this special case all envy-free allocations have the same utility.

The problem of determining worst case examples, i.e., $p_{envy}(n,n)$ can be formulated as an integer linear programming problem:
\begin{eqnarray}
  \max \sum_{i=1}^n\sum_{j=1}^n z_{ij}-\alpha\sum_{i=1}^n x_{ii}\label{tf}\\
  x_{ij}\in \mathbb{R}_{\ge 0}\quad \forall \, 1\le i,j\le n &\quad\quad&
  \sum_{i=1}^n x_{ij}=1 \quad \forall\, 1\le j\le n\label{ie_ut}\\
  x_{jj}\ge x_{ij} \quad \forall\, 1\le i,j\le n\label{ie_ef}\\
  y_{ij}\in \{0,1\}\quad\forall \, 1\le i,j\le n &\quad\quad&
  \sum_{j=1}^n y_{ij}=1\quad\forall\, 1\le i\le m\label{ie_oa}\\
  z_{ij}\in \mathbb{R}_{\ge 0}\quad\forall \, 1\le i,j\le n &\quad\quad&
  z_{ij}\le \min\!\left(y_{ij},x_{ij}\right)\quad\forall \, 1\le i,j\le n\label{ie_gw}
\end{eqnarray}
Here inequalities (\ref{ie_ut}) specify the non-negative utilities of agent $j$ for item $i$, which sum up to one.
The envy-freeness of the allocation given by $A_j=\{j\}$ is guaranteed by Inequality~(\ref{ie_ef}). In an optimal assignment
item $i$ is assigned to agent $j$ iff $y_{ij}=1$, see inequalities (\ref{ie_oa}). The auxiliary variables 
$z_{ij}$ measure the contribution to the global welfare, see inequalities (\ref{ie_gw}). If the target function 
(\ref{tf}) admits a non-negative value, then we have $p_{envy}(n,n)\ge \alpha$ and $p_{envy}(n,n)< \alpha$ otherwise. We can 
already conclude that the suppremum is attained in the definition for the price of fairness.

Using a bisection approach we were able to exactly determine $p_{envy}(n,n)$ for all $n\le 9$, i.e., $p_{envy}(n,n)=1,1,\frac{8}{7},\frac{4}{3}, 
\frac{60}{43},\frac{3}{2},\frac{63}{40},\frac{72}{43},\frac{9}{5}$. It turned out that the optimal solution 
for $n\ge 2$ have a rather special structure. The $x_{ij}$ all were either equal to zero or to $\frac{1}{k_j}$, where $2\le k_j\le n$ is an integer.
Even more, at most three different $k_j$-values are attained for a fixed number $n$, where one case is always $k_j=n$. In the next subsection we 
theoretically prove this empirical observation.

\subsection{Special structure of the optimal solutions for $\mathbf{m=n}$}
For the ease of notation we use $\tau:\{1,\dots,n\}\rightarrow\{1,\dots,n\}$, mapping an item $i$ to an agent $j$, representing an
optimal assignment, i.e., $y_{i\tau(i)}=1$ for all $1\le i\le n$. By $u^\star(x)=\sum_{i=1}^n x_{i\tau(i)}$ we denote 
the welfare of an optimal assignment and by $u^\star_f(x)=\sum_{i=1}^{n} x_{ii}$ the welfare of an optimal envy-free assignment. 
In the following we always assume that $x$ represents utilities from an example attaining $p_{envy}(n,n)$. We call 
an agent $j$ \textit{big} if $j\in \operatorname{im}(\tau)$ and \textit{small} otherwise.

\begin{lemma}
  If agent $j$ is small, then we have $x_{ij}=\frac{1}{n}$ for all $1\le i\le n$. 
\end{lemma}
\begin{proof}
  If $x_{jj}\le \frac{1}{n}$ then we have $x_{ij}=\frac{1}{n}$ for all $1\le i\le n$ so that we assume $x_{jj}> \frac{1}{n}$.
  Consider $x'$ arising from $x$ by setting $x'_{ij}=\frac{1}{n}$ for all $1\le i\le n$. With this we have 
  $u^\star_f(x')<u^\star_f(x)$ and $u^\star(x')\ge u^\star(x)$.
\end{proof}

\begin{lemma}
  If agent $j$ is big, then we have $x_{ij}=x_{jj}$ for all $i$ with $\tau(i)=j$. 
\end{lemma}
\begin{proof}
  We set $w=\sum_{i:\tau(i)=j} x_{ij}$ and $k=\left|\left\{i\mid \tau(i)=j\right\}\right|$. W.l.o.g.\ assume 
  $x_{jj}>\frac{w}{k}$. Consider $x'$ arising from $x$ by setting $x'_{ij}=\frac{w}{k}$ for all $1\le i\le n$ with 
  $\tau(i)=j$. With this we have $u^\star_f(x')<u^\star_f(x)$ and $u^\star(x')\ge u^\star(x)$.
\end{proof}

\begin{lemma}
  If agent $j$ is big, then we can assume $x_{ij}=0$ or $x_{ij}=\frac{1}{n}$ for all $i$ with $\tau(i)\neq j$ w.l.o.g. 
\end{lemma}
\begin{proof}
  Again we set $w=\sum_{i:\tau(i)=j} x_{ij}$ and $k=\left|\left\{i\mid \tau(i)=j\right\}\right|$, where we can assume $k<n$. As an 
  abbreviation we use $a=u^\star(x)-w$ and $b=u^\star_f(x)-\frac{w}{k}$. With this we define $f(t)=\frac{a+k\cdot 
  \left(\frac{1}{k}-t\right)}{b+\frac{1}{k}-t}$ for $0\le t\le \frac{1}{k}-\frac{1}{n}$. At $t=\frac{1-w}{k}$ the function $f$ gives
  the price of fairness for $x$. In general we consider $x'(t)$ arising from $x$ by setting $x'_{ij}(t)=\frac{1}{k}-t$ for all $i$ 
  with $\tau(i)=j$ and $x'_{ij}(t)=\frac{kt}{n-k}$ otherwise. We have $u^\star_f(x'(t))=b+\frac{1}{k}-t$ and $u^\star(x'(t))\ge 
  a+k\cdot\left(\frac{1}{k}-t\right)$ 
  so that the price of fairness for $x'(t)$ is at least $f(t)$. The construction is feasible for $0\le t\le \frac{1}{k}-\frac{1}{n}$ only, 
  since otherwise $x'_{jj}\ge x'_{ij}$ is violated. Thus $f(t)$ is well defined and we have 
  $f'(t)=(-bk+a)/(b+1/k-t)^2$.
  So either $f(t)$ is strictly monotonic or constant and attains its maximum at the boundary.
\end{proof}

Thus we can assume w.l.o.g.\ that $x_{\star j}$ consists of zeros and $k_j$ times the entry $1/k_j$, where $k_j$ is a positive integer. 
If $k_j<n$, then all $k_j$ items with utility $1/k_j$ are assigned to agent $j$ in an optimal solution. If $\tau(i)\neq j$, then 
$x_{ij}\in\{0,1/n\}$. We can further assume that there is at most one big agent $j$ with $k_j=n$.

\subsection{An improved ILP formulation and an almost tight bound for $\mathbf{p_{envy}(n,n)}$}
Given the structural result from the previous subsection we can reformulate the ILP to:
\begin{eqnarray*}
    \!\!\!\!\!\!&&\max \sum_{i=1}^n \frac{r_i}{i}-\alpha\sum_{i=1}^n \frac{s_i}{i}\\
    \!\!\!\!\!\!&&s_{i}\in \mathbb{Z}_{\ge 0}\quad \forall \, 1\le i\le n \quad\quad\quad  \sum_{i=1}^n s_i=n  \\
    \!\!\!\!\!\!&&r_{i}\in \mathbb{Z}_{\ge 0}\quad \forall \, 1\le i\le n \quad\quad\quad \sum_{i=1}^n r_i=n  \quad\quad\quad
    r_i\le i\cdot s_i\quad \forall\,1\le i\le n
\end{eqnarray*}
Here $s_i$ counts how often we have $k_j=\frac{1}{i}$ and $r_j$ counts how often we have $x_{i\tau(i)}=\frac{1}{j}$. Having this ILP 
formulation at hand the exact values of $p_{envy}(n,n)$ can be computed easily for all $n\le 100$. We observe that in each case 
at most three values of the vector $s$ are non-zero -- going in line with our previous empirical findings. 

\begin{lemma}
  If $x_{jj}=\frac{1}{k}$ and $x_{j'j'}=\frac{1}{k+g}$, where $k,g\in\mathbb{N}$ and $k+g<n$, then $g\le 1$.
\end{lemma}
\begin{proof}
  First note that $j$ and $j'$ are big agents. To the contrary assume $g\ge 2$ and consider $x'$ arising from $x$ by replacing $x_{ij}$ 
  by elements of the form $\frac{1}{k+1}$, $0$ and
  $x_{ij'}$ by elements of the form $\frac{1}{k+g-1}$, $0$ in a suitable way. Since
  $$
    \frac{1}{k}\cdot k+\frac{1}{k+g}\cdot(k+g)=2=\frac{1}{k+1}\cdot(k+1)+\frac{1}{k+g-1}\cdot(k+g-1)
  $$
  we have $u^\star(x')\ge u^\star(x)$.
  Since
  $$
    \frac{1}{k}+\frac{1}{k+g}=\frac{2k+g}{k^2+kg}>\frac{2k+g}{k^2+kg+g-1}=\frac{1}{k+1}+\frac{1}{k+g-1}
  $$
  we have $u^\star_f(x')< u^\star_f(x)$.
\end{proof}

\begin{theorem}
  $p_{envy}(n,n)\le\frac{1}{2}\sqrt{n}+O(1)$.
\end{theorem}
\begin{proof}
  Choose $k$ such that $k_j\in\{k-1,k,n\}$ for all $j\in \mathcal{J}$ and set $a=\left|\{j\mid k_j=k\}\right|$, 
  $b=\left|\{j\mid k_j=k-1\}\right|$. With this we have $u^\star(x)=a+b+c/n$, where $c=n-ak-b(k-1)$ and 
  $u^\star_f(x)=a/k+b/(k-1)+(n-a-b)/n$. Next we set $d=a+b$ and $\tilde{k}=(ak+b(k-1))/d$. Since $c/n\le 1$, 
  $d/\tilde{k}\le a/k+b/(k-1)$, and $\tilde{k}\le n/d$ we have
  $$
    \frac{u^\star(x)}{u^\star_f(x)}\le \frac{d+1}{\frac{d}{\tilde{k}}+\frac{n-d}{n}}\le \frac{n(d+1)}{d^2+n-d}\le \frac{n(d+1)}{d^2+n}=:g(d).
  $$
  For $d\in \{0,n\}$ we have $g(d)=1$. The unique local maximum of $g(d)$ in $(0,n)$ is at attained at $d=-1+\sqrt{1+n}$. 
  Thus $p_{envy}(n,n)=\frac{u^\star(x)}{u^\star_f(x)}\le g(d)\le\max\left(1,\frac{1}{2}\sqrt{n}+\frac{1}{n}+1\right)$. 
\end{proof}

\begin{lemma}
  $p_{envy}(n,n)\ge\frac{1}{2}\sqrt{n}-\frac{1}{2}$.
\end{lemma}
\begin{proof}
  Set $a=k=\left\lfloor\sqrt{n}\right\rfloor$ and $x'$ with $a$ rows of the form $(\frac{1}{k},\dots,\frac{1}{k},0,\dots,0)$ and 
  $n-a$ rows of the form $(\frac{1}{n},\dots,\frac{1}{n})$. With this we have
  $$
    p_{envy}(n,n)\ge \frac{u^\star(x')}{u^\star_f(x')}\ge \frac{a+\frac{n-ak}{n}}{\frac{a}{k}+\frac{n-a}{n}}
  \ge\frac{a}{2}\ge\frac{\sqrt{n}}{2}-\frac{1}{2}.
  $$ 
\end{proof}

Thus we can state $p_{envy}(n,n)=\frac{1}{2}\sqrt{n}+\Theta(1)$.

\section{Bounds for $\mathbf{p_{envy}(n,m)}$ for a small number of items}
If the number of items is not too large, i.e., $m\le n+c\sqrt{n}$ for a constant $c$, then we can 
utilize our results for the case $m=n$ in order to deduce an $\Theta(\sqrt{n})$-bound for the price
of fairness.

\begin{theorem}
  If $m\in n+\Theta(\sqrt{n})$ with $m\ge n$ then $p_{envy}(n,m)\in \Theta(\sqrt{n})$.
\end{theorem}
\begin{proof}
  Consider a utility matrix $x$ with $p_{envy}(n,m)\le u^\star(x)/u^\star_f(x)+\epsilon$ for a small constant $\epsilon\ge 0$. Choose 
  a constant $c\in\mathbb{R}_{\ge 0}$ with $m=c+\sqrt{n}$. By $S\subseteq\mathcal{J}$ we denote the set of agents to which a single 
  item is assigned in the optimal envy-free allocation and set $s=|S|$. All other agents get at least two items so that $n-s\le 2c\sqrt{n}$ 
  and $s\ge n-c\sqrt{n}$. Now consider another utility matrix $x'$ arising from $x$ as follows. For each agent in $S$ copy the utility row 
  from $x$. Replace the remaining agents from $\mathcal{J}\backslash S$ by $m-s\ge n-s$ new agents having utility $1/m$ for each item. 
  With this we have $u^\star_f(x)\ge u^\star_f(x')-(m-s)\cdot\frac{1}{m}\ge u^\star_f(x')-\frac{3c}{\sqrt{n}}$ and 
  $u^\star(x)\le u^\star(x')+2c\sqrt{n}$ since each agent $j\notin S$ could contribute at most $1$ to $u^\star(x)$. Thus we have
  $$
    \frac{u^\star(x)}{u^\star_f(x)}\le\frac{u^\star(x')+2c\sqrt{n}}{u^\star_f(x')-\frac{3c}{\sqrt{n}}}
    \le \frac{1}{1-\frac{3c}{\sqrt{n}}} \cdot\left(\frac{u^\star(x')}{u^\star_f(x')}+2c\sqrt{n}\right)
  $$
  due to $u^\star_f(x')\ge 1$. Since the number of agents coincides with the number of items in $x'$, the right hand side
  of the last inequality is in $O(\sqrt{n})$. The lower bound follows from the case $m=n$.
\end{proof}

\section{Conclusion}
We have introduced the price of fairness in terms of the number of agents and the number of items. As a special case we have considered 
the allocation of indivisible goods with respect to envy-freeness as a fairness criterion and \textit{normalized} additive 
utility functions. It turned out that the price of fairness is significantly lower if only a small number of items has to be allocated
compared to the case of a large number of items. Up to a constant we have determined the exact value of the price of fairness for the special 
case when the number of items coincides with the number of agents. In order to determine the exact value we have given an efficient ILP 
formulation.

We close with some open questions: Can further values of $p_{envy}(n,m)$, where $m>n$, be computed exactly? Can the ILP approach be extended to $m>n$? 
What is the price of fairness in our setting for $m\in\Theta(n)$ (or more generally, for $m\in \Theta(n^\alpha)$ with $\alpha<2$)? What happens 
for other fairness criteria? 


\begin{thebibliography}{1}

\bibitem{bertsimas2011price}
D.~Bertsimas, V.F. Farias, and N.~Trichakis, \emph{The price of fairness},
  Operations research \textbf{59} (2011), no.~1, 17--31.

\bibitem{brams1996fair}
S.J. Brams and A.D. Taylor, \emph{Fair division: From cake-cutting to dispute
  resolution}, Cambridge University Press, 1996.

\bibitem{caragiannis2012efficiency}
I.~Caragiannis, C.~Kaklamanis, P.~Kanellopoulos, and M.~Kyropoulou, \emph{The
  efficiency of fair division}, Theory of Computing Systems \textbf{50} (2012),
  no.~4, 589--610.
\end{thebibliography}

\end{document}